%% file: arxiv.tex
\documentclass[11pt]{article}

\pdfoutput=1

\usepackage{a4}
\usepackage[top=30truemm,bottom=30truemm]{geometry}

\usepackage{authblk}
\usepackage{graphicx}
\usepackage{rotating}
\usepackage{comment}
\usepackage{multirow}

\usepackage{amsmath}
\usepackage{amssymb}
\usepackage{algorithm}
\usepackage[noend]{algpseudocode}
\usepackage{algorithmicx}
\usepackage{url}
\usepackage{enumerate}
\usepackage{bm}
\usepackage{setspace}
\usepackage{epsfig}
\usepackage{color}

\usepackage[final,mode=multiuser]{fixme}

\fxuseenvlayout{colorsig}
\fxusetargetlayout{color}

\FXRegisterAuthor{an}{aan}{AN}
\FXRegisterAuthor{yn}{ayn}{YN}
\FXRegisterAuthor{si}{asi}{SI}
\FXRegisterAuthor{hb}{ahb}{HB}
\FXRegisterAuthor{mt}{amt}{MT}
\fxsetface{env}{}

\newcommand{\rle}[1]{\mathsf{rle}(#1)}
\newcommand{\dtw}{\mathsf{dtw}}
\newcommand{\ddtw}{\mathbf{D^2TW}}

\newcommand{\Tlist}{\Delta_{\mathrm{T}}}
\newcommand{\Llist}{\Delta_{\mathrm{L}}}
\newcommand{\Blist}{\Delta_{\mathrm{B}}}
\newcommand{\Rlist}{\Delta_{\mathrm{R}}}

\newcommand{\changed}{\#_{\mathrm{chg}}}
\newcommand{\cost}{\mathsf{w}}

\newcommand{\ITOP}{i_{\mathrm{T}}^I}
\newcommand{\IBOT}{i_{\mathrm{B}}^I}
\newcommand{\JRIGHT}{j_{\mathrm{R}}^J}
\newcommand{\JLEFT}{j_{\mathrm{L}}^J}

\newcommand{\D}{\mathit{D}}
\newcommand{\B}{\mathcal{B}}
\newcommand{\Dp}{\mathit{D^\prime}}
\newcommand{\Bp}{\mathit{B^\prime}}

\newcommand{\DR}{\mathit{DR}}
\newcommand{\DRp}{\mathit{DR^\prime}}
\newcommand{\DS}{\mathit{DS}}
\newcommand{\DSp}{\mathit{DS^\prime}}

\newtheorem{theorem}{Theorem}
\newtheorem{lemma}{Lemma}
\newtheorem{corollary}{Corollary}

\newenvironment{proof}{\begin{trivlist} \item{\itshape Proof.}}{\end{trivlist}}

\newcommand{\qed}{\hfill $\Box$}

\title{Towards Efficient Interactive Computation of \\
  Dynamic Time Warping Distance}

\author{Akihiro~Nishi$^1$}
\author{Yuto~Nakashima$^1$}
\author{Shunsuke~Inenaga$^{1,2}$}
\author{Hideo~Bannai$^3$}
\author{Masayuki~Takeda$^1$}

\affil{
  \textit{$^1$ Department of Informatics, Kyushu University, Fukuoka, Japan}\\
  \textit{$^2$ PRESTO, Japan Science and Technology Agency, Kawaguchi, Japan}\\
  \textit{$^3$ M\&D Data Science Center, Tokyo Medical and Dental University, Tokyo, Japan}
}

\date{}

\begin{document}
\maketitle

\begin{abstract}
  The \emph{dynamic time warping} (\emph{DTW}) is a widely-used
  method that allows us to efficiently compare two time series that can vary in speed.
  Given two strings $A$ and $B$ of respective lengths $m$ and $n$,
  there is a fundamental dynamic programming algorithm that computes
  the DTW distance $\dtw(A,B)$ for $A$ and $B$
  together with an optimal alignment in $\Theta(mn)$ time and space.
  In this paper, we tackle the problem of interactive computation
  of the DTW distance for dynamic strings,
  denoted $\ddtw$, where character-wise edit operation
  (insertion, deletion, substitution) can be performed at
  an \emph{arbitrary} position of the strings.
  Let $M$ and $N$ be the sizes of the \emph{run-length encoding} (\emph{RLE})
  of $A$ and $B$, respectively.
  We present an algorithm for $\ddtw$ that occupies $\Theta(mN+nM)$ space
  and uses $O(m+n+\changed) \subseteq O(mN + nM)$ time to update
  a compact differential representation $\DS$ of the DP table per edit operation,
  where $\changed$ denotes the number of cells in $\DS$ whose values
  change after the edit operation.
  Our method is at least as efficient as
  the algorithm recently proposed by Froese et al.
  running in $\Theta(mN + nM)$ time,
  and is faster when $\changed$ is smaller than $O(mN + nM)$
  which, as our preliminary experiments suggest,
  is likely to be the case in the majority of instances.
\end{abstract}


\input{introduction}

\input{preliminaries}

\input{algorithm}
\input{evaluation}

\input{ack}

\bibliographystyle{abbrv}
\bibliography{ref}

\clearpage
\appendix

\input{appendix}

\end{document}

%% file: introduction.tex
\section{Introduction}
The \emph{dynamic time warping} (\emph{DTW}) is a classical and widely-used
method that allows us to efficiently compare two temporal sequences or time series that can vary in speed.
A fundamental dynamic programming algorithm computes
the DTW distance $\dtw(A,B)$ for two strings $A$ and $B$
together with an optimal alignment
in $\Theta(mn)$ time and space~\cite{Sakoe1978}, where $|A| = m$ and $|B| = n$.
This algorithm allows one to update the DP table $\D$ for $\dtw(A,B)$
in $O(m)$ time (resp. $O(n)$ time) when a new character is appended to $B$ (resp. to $A$).

In this paper, we introduce the ``dynamic'' DTW problem,
denoted $\ddtw$, where character-wise edit operation
(insertion, deletion, substitution) can be performed at
an \emph{arbitrary} position of the strings.
More formally, we wish to maintain a (space-efficient) representation of $\D$
that can dynamically be modified according to a given operation.
This representation should be able to quickly answer the value of $\D[m,n] = \dtw(A, B)$
upon query,
together with an optimal alignment achieving $\dtw(A, B)$. 
This kind of interactive computation for (a representation of) $\D$
can be of practical merits, e.g. when simulating stock charts, or editing musical sequences.
Another example of applications of $\ddtw$ is a sliding window version of DTW
which computes $\dtw(A, B[j..j+d-1])$ between $A$ and
every substring $B[j..j+d-1]$ of $B$ of arbitrarily fixed length $d$.

\emph{Incremental/decremental} computation of a DP table
is a restricted version of the aforementioned interactive computation,
which allows for prepending a new character to $B$,
and/or deleting the leftmost character from $B$.
A number of incremental/decremental computation algorithms
have been proposed for the \emph{unit-cost edit distance} and \emph{weighted edit distance}:
Kim and Park~\cite{kim_park_jda2004} showed
an incremental/decremental algorithm for the unit-cost edit distance
that occupies $\Theta(mn)$ space and runs in $O(m+n)$ time per operation.
Hyyr\"o et al.~\cite{HyyroNI10} proposed 
an algorithm for the edit distance with integer weights
which uses $\Theta(mn)$ space and runs in $O(\min\{c(m+n), mn\})$ time per operation,
where $c$ is the maximum weight in the cost function.
This translates into $O(m + n)$ time under constant weights.
Schmidt~\cite{schmidt98} gave an algorithm that uses $\Theta(mn)$ space
and runs in $O(n \log m)$ time per operation for a general weighted edit distance.
Hyyr\"o and Inenaga~\cite{HyyroI16} presented
a space efficient alternative to incremental/decremental
unit-cost edit distance computation which runs in $O(m+n)$ time per operation
but uses only $\Theta(mN + nM)$ space, 
where $M$ and $N$ are the sizes of \emph{run-length encoding} (RLE) of $A$ and $B$, respectively.
Since $M \leq m$ and $N \leq n$ always hold,
the $mN + nM$ terms can be much smaller than the $mn$ term
for strings that contain many long character runs.
Later, Hyyr\"o and Inenaga~\cite{HyyroI18} presented a space-efficient alternative
for edit distance with integer weights,
which runs in $O(\min\{c(m+n),mn\})$ time per operation and requires $\Theta(mN + nM)$ space.

Fully-dynamic interactive computation for
the (weighted) edit distance was also considered:
Let $j^*$ be the position in $B$ where the modification
has been performed.
For the unit cost edit distance,
Hyyr\"o et al.~\cite{Hyyro2015} presented a representation of the DP table 
which uses $\Theta(mn)$ space and can be updated in $O(\min\{rc(m+n),mn\})$ time per operation, 
where $r = \min\{j^*, n-j^*+1\}$ and
$c$ is the maximum weight.
They also showed that there exist instances
that require $\Omega(\min\{rc(m+n),mn\})$ time to 
update their data structure per operation.
Very recently, Charalampopoulos et al.~\cite{CharalampopoulosKM20}
showed how to maintain
an optimal (weighted) alignment of two fully-dynamic strings
in $\tilde{O}(n \min\{\sqrt{n}, c\})$ time per operation,
where $m = n$.

While computing \emph{longest common subsequence} (\emph{LCS})
and weighted edit distance of strings of length $n$ 
can both be reduced to computing DTW of strings of length $O(n)$~\cite{AbboudBW15,Kuszmaul19},
a reduction to the other direction is not known.
It thus seems difficult to directly apply any of the aforementioned algorithms
to our $\ddtw$ problem.
Also, a conditional lower bound suggests
that strongly sub-quadratic DTW algorithms are unlikely to exist~\cite{AbboudBW15,BringmannK15}.
Thus, any method that recomputes the na\"ve DP table $\D$ from scratch should take
almost quadratic time per update.

\vspace*{0.5pc}
\noindent \textbf{Our contribution.}
This paper takes the first step towards an efficient solution to $\ddtw$.
Namely, we present an algorithm for $\ddtw$ that occupies $\Theta(mN+nM)$ space
and uses $O(m+n+\changed)$ time to update
a compact differential representation $\DS$ for the DP table $\D$
per edit operation,
where $\changed$ denotes the number of cells in $\DS$ whose values
change after the edit operation.
Since $\changed = O(mN+nM)$ always holds,
our method is always at least as efficient as the na\"ive method
that recomputes the full DP table $\D$ in $\Theta(mn)$ time,
or the algorithm of Froese et al.~\cite{FroeseJRW2020} that recomputes
another sparse representation of $\D$ in $\Theta(mN + nM)$ time.
While there exist worst-case instances that give $\changed = \Omega(mN+nM)$,
our preliminary experiments suggest that, in many cases,
$\changed$ can be much smaller than the size of $\DS$ which is $\Theta(mN + nM)$.

Technically  our algorithm is most related to 
Hyyr\"o et al.'s method~\cite{HyyroNI10,Hyyro2015} and
Froese et al.'s method~\cite{FroeseJRW2020},
but our algorithm is not straightforward from these.

%% file: preliminaries.tex
\section{Preliminaries} \label{sec:prelim}

We consider sequences (strings) of characters
from an alphabet $\Sigma$ of real numbers.
Let $A = a_1, \ldots, a_m$ be a string consisting of $m$ characters
from $\Sigma$.
The {\em run-length encoding} $\rle{A}$ of string $A$ is
a compact representation of $A$ such that
each maximal run of the same characters in $A$ is represented
by a pair of the character and the length of the run.
More formally, let $\mathbb{N}$ denote the set of positive integers.
For any non-empty string $A$,
$\rle{A} = a_1^{e_1} \cdots a_M^{e_M}$,
where $a_I \in \Sigma$ and $e_I \in \mathbb{N}$ for any $1 \leq I \leq M$,
and $a_I \neq a_{I+1}$ for any $1 \leq I < M$.
Each $a_I^{e_I}$ in $\rle{A}$ is called a (character) \emph{run},
and $e_I$ is called the exponent of this run.
The \emph{size} of $\rle{A}$ is the number $M$ of runs in $\rle{A}$.
E.g., for string $A = \mathtt{aacccccccbbabbbb}$ of length 16,
$\rle{A} = \mathtt{a}^2 \mathtt{c}^7 \mathtt{b}^2 \mathtt{a}^1 \mathtt{b}^4$ and its size is 5.

\emph{Dynamic time warping} (\emph{DTW})
is a commonly used method to compare two temporal sequences
that may vary in speed.
Consider two strings $A = a_1, \ldots, a_m$ and
$B = b_1, \ldots, b_n$.
To formally define the DTW for $A$ and $B$,
we consider an $m \times n$ grid graph $\mathcal{G}_{m, n}$ such that 
each vertex $(i, j)$ has (at most) three directed edges;
one to the lower neighbor $(i+1, j)$ (if it exists),
one to the right neighbor $(i, j+1)$ (if it exists),
and one to the lower-right neighbor $(i+1, j+1)$ (if it exists).
A path in $\mathcal{G}_{m,n}$ that starts from vertex $(1, 1)$ and ends at vertex $(m, n)$
is called a \emph{warping path},
and is denoted by
a sequence $(1, 1), \ldots, (i, j), \ldots, (m, n)$ of adjacent vertices.
Let $\mathcal{P}_{m, n}$ be the set of all warping paths in $\mathcal{G}_{m, n}$.
Note that each warping path in $\mathcal{P}_{m, n}$ corresponds to
an alignment of $A$ and $B$.
The DTW for strings $A$ and $B$, denoted $\dtw(A, B)$, is defined by
$ \dtw(A, B) = \min_{p \in \mathcal{P}_{m, n}} \sqrt{\sum_{(i, j) \in p}(a_i - b_j)^2}$.

The fundamental $\Theta(mn)$-time and space solution
for computing $\dtw(A,B)$, given in~\cite{Sakoe1978},
fills an $m \times n$ dynamic programming table $\D$ such that
$\D[i,j] = \dtw(A[1..i],B[1..j])^2$ for $1 \leq i \leq m$ and $1 \leq j \leq n$.
Therefore, after all the cells of $\D$ are filled,
the desired result $\dtw(A,B)$ can be obtained by $\sqrt{\D[m,n]}$.
The value for each cell $\D[i,j]$ is computed
by the following well-known recurrence:
\begin{equation}
\begin{array}{lll}  
  \D[1,1] & = & (a_1 - b_1)^2, \\
  \D[i,1] & = & \D[i-1,1] + (a_i - b_1)^2 \quad \text{for }  1 < i \leq m, \\
  \D[1,j] & = & \D[1,j-1] + (a_1 - b_j)^2 \quad \text{for }  1 < j \leq n, \\
  \D[i,j] & = & 
      \begin{array}[t]{l}
       \min\{\D[i,j-1], \D[i-1,j], \D[i-1,j-1] \} + (a_i-b_j)^2\\
       \ \text{for $1 < i \leq m$ and $1 < j \leq n$.}
      \end{array}
\end{array} \label{recurrence}
\end{equation}

In the rest of this paper, we will consider the problem of maintaining 
a representation for $\D$,
each time one of the strings, $B$, is dynamically modified
by an edit operation (i.e. single character insertion, deletion, or substitution)
on an \emph{arbitrary} position in $B$.
We call this kind of interactive computation of $\dtw(A,B)$ as
the \emph{dynamic} DTW computation, denoted by $\ddtw$.

Let $\Bp$ denote the string after an edit operation is performed on $B$,
and $\Dp$ denote the dynamic programming table $\D$ after it has been updated to correspond to $\dtw(A, \Bp)$.
In a special case where the edit operation is performed at the right end of $B$,
where we have $\Bp=Bc$ (insertion),
$\Bp=B[1..n-1]$ (deletion) or $\Bp=B[1..n-1]c$ (substitution) with
a character $c \in \Sigma$,
then $\D$ can easily be updated to $\Dp$ in $O(m)$ time by simply computing
a single column at index $j = n$ or $j = n+1$ using recurrence~(\ref{recurrence}).

\begin{figure}[t]
    \centering
    \includegraphics[clip,scale=0.5,bb=0 0 581 222]{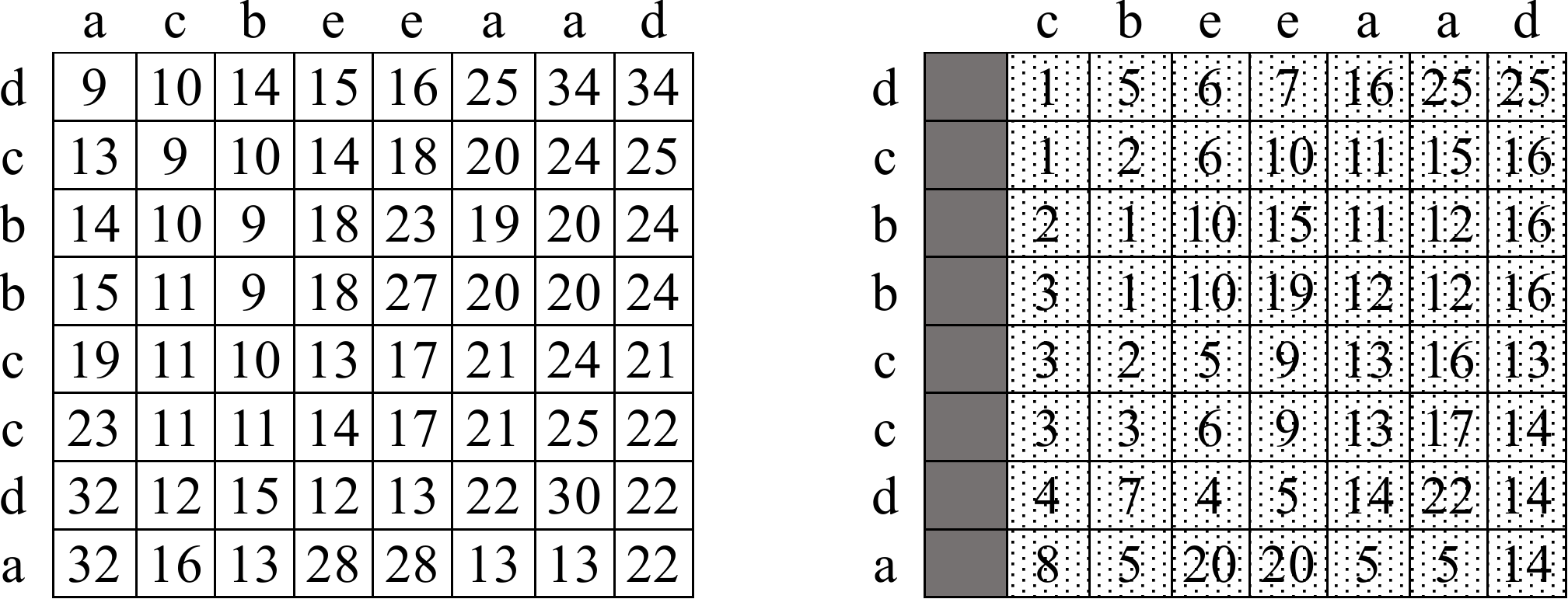}
    \caption{In this example
      where $A = \textrm{dcbbccda}$ and $B = \textrm{acbeeaad}$,
      the values of $\Theta(mn)$ cells of the DP table for $\dtw(A,B)$ change after the edit operation on $B$ (here, the first character $B[1] = \mathrm{a}$ of $B$ was deleted).}
    \label{fig:worst_case}
\end{figure}

As in Figure~\ref{fig:worst_case},
in the worst case, the values of $\Theta(mn)$ cells of the DP table for $\dtw(A,B)$ can change after an edit on $B$.
The following lemma gives a stronger statement
that updating $\D$ to $\Dp$ in our $\ddtw$ scenario cannot be amortized:
\begin{lemma} \label{lem:updating_DP-table_worst-case}
  There are strings $A$, $B$ and a sequence of
  $k$ edits on $B$ such that
  $\Theta(kmn)$ cells in $\Dp$ have different values in the corresponding cells in $\D$.
\end{lemma}

%% file: algorithm.tex
\section{Our $\ddtw$ Algorithm based on RLE} \label{sec:algorithm}

We first explain the data structures which are used in our algorithm.

\vspace*{1pc}
\noindent \textbf{Differential representation $\DR$ of $\D$.}
The first idea of our algorithm
is to use a differential representation $\DR$ of $\D$:
Each cell of $\DR$ contains two fields
that respectively store the horizontal difference
and the vertical difference, namely,
$\DR[i,j].U = \D[i,j] - \D[i-1,j]$ and $\DR[i,j].L = \D[i,j] - \D[i,j-1]$.
We let $\DR[i,1].L = 0$ for any $1 \leq i \leq m$ and 
  $\DR[1,j].U = 0$ for any $1 \leq j \leq n$.
The diagonal difference $\D[i,j] - \D[i-1,j-1]$ can easily be computed
from $\DR[i,j].U$ and $\DR[i,j].L$ and thus is not explicitly stored in $\DR[i,j]$.

In our algorithm we make heavy use of the following lemma:
\begin{lemma} \label{lem:DR_recursive}
  For any $1 < i \leq m$,
  \[
  \DR[i,j].U =
  \begin{cases}
    (a_i-b_1)^2 & \mbox{if } j = 1, \\
    z - \DR[i-1,j].L & \mbox{if } 2 \leq j \leq n,
  \end{cases}
  \]
  and for any $1 < j \leq n$,
  \[
  \DR[i,j].L =
  \begin{cases}
    (a_1-b_j)^2 & \mbox{if } i = 1, \\
    z - \DR[i,j-1].U & \mbox{if } 2 \leq i \leq m,
  \end{cases}
  \]
where $z = \min\{\DR[i-1,j].L, \ \DR[i,j-1].U, \ 0\} + (a_i-b_j)^2$.
\end{lemma}

\begin{proof}
  $\DR[i,1].U = (a_i-b_1)^2$ and $\DR[1,j].L = (a_1-b_j)^2$
  are clear from recurrence~(\ref{recurrence}).
  Now we consider $1 < i \leq m$ and $1 < j \leq n$,
  and let $d = \D[i-1,j-1]$, $x = \DR[i-1,j].L$, $y = \DR[i,j-1].U$, and $d + z = D[i,j]$.
  Then we have $\D[i-1,j] = d+x$ and $\D[i,j-1] = d+y$ (see Figure~\ref{lem:DR_recursive}).
  It follows from the definition of $\DR$
  that $\DR[i,j].U = \D[i,j] - \D[i-1,j] = z-x$ and
  $\DR[i,j].L = \D[i,j] - \D[i,j-1] = z-y$.
  Since $\D[i,j] = \min\{\D[i-1,j-1], \D[i-1,j], \D[i,j-1]\} + (a_i-b_j)^2$
  by recurrence~(\ref{recurrence}), we obtain $d+z = \min\{d,d+x,d+y\}+(a_i-b_j)^2$
  which leads to $z = \min\{x,y,0\}+(a_i-b_j)^2$.
  \qed
\end{proof}

\noindent \textbf{RLE-based sparse differential representation $\DS$.}
The second key idea of our algorithm is 
to divide the dynamic programming table $\D$ into ``boxes'' that are defined by intersections of maximal runs of $A$ and $B$.
Note that $\D$ contains $M \times N$ such boxes.
Let $\rle{A} = A_1^{k_1}\dots A_M^{k_M}$ and
$\rle{B} = B_1^{l_1} \dots B_N^{l_N}$ be the RLEs of $A$ and $B$.
Let $\ITOP=\sum_i^{I-1}{k_i}+1$,
$\IBOT=\sum_i^I{k_i}$,
$\JLEFT=\sum_j^{J-1}{l_j}+1$, and
$\JRIGHT=\sum_j^J{l_j}$.
We define a sparse table $\DS$ for $\DR$
that consists only of the rows and columns on
the borders of the maximal runs in $A$ and $B$.
Namely, $\DS$ is a sparse table that only stores
the rows $\ITOP,\IBOT~(1 \leq I \leq M)$
and the columns $\JLEFT,\JRIGHT~(1 \leq J \leq N)$,
of $\DR$ (see Figure~\ref{fig:DS}).
\begin{figure}[tb]
  \centering
  \begin{tabular}{c}

    \begin{minipage}{0.37\hsize}
      \centering
      \includegraphics[clip,scale=0.4,bb=0 0 266 240]{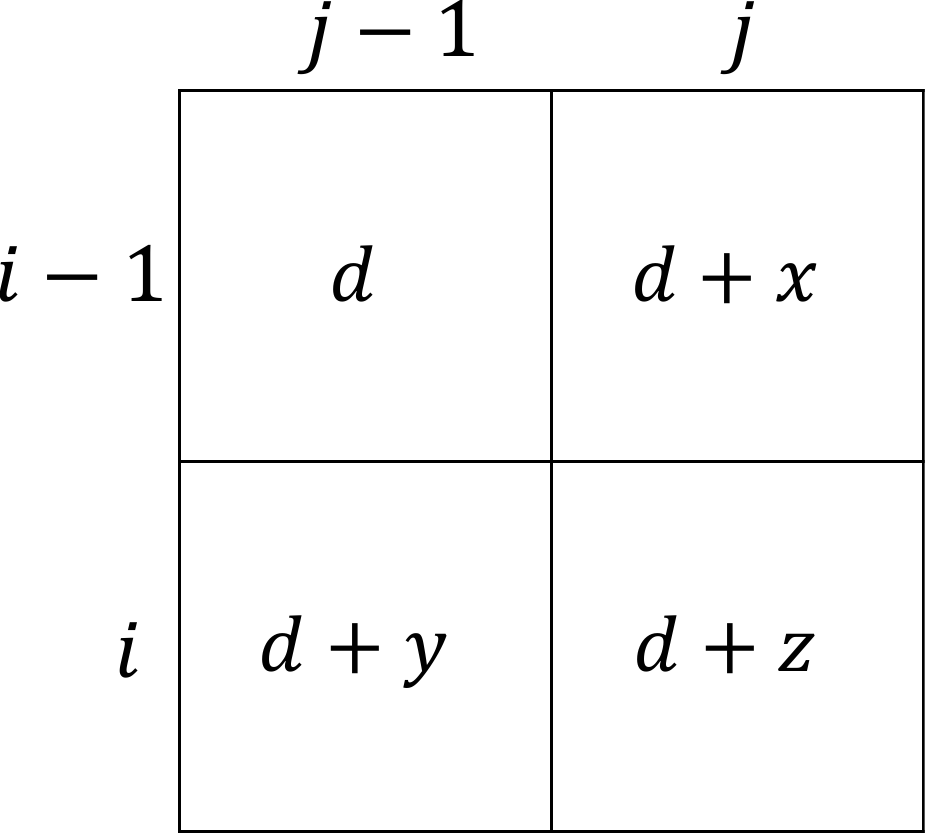}
    \caption{Illustration for Lemma~\ref{lem:DR_recursive}
    which depicts the corresponding cells of the dynamic programming table $\D$,
    where $\D[i-1,j-1] = d$, $\D[i-1,j] = d + x$, $\D[i,j-1] = d + y$, and $D[i,j] = d + z$.}
    \label{fig:DR_recursive}
   \end{minipage}

    \begin{minipage}{0.03\hsize}
        \hfill
    \end{minipage}

    \begin{minipage}{0.57\hsize}
        \centering
        \includegraphics[clip,scale=0.5,bb=0 0 475 293]{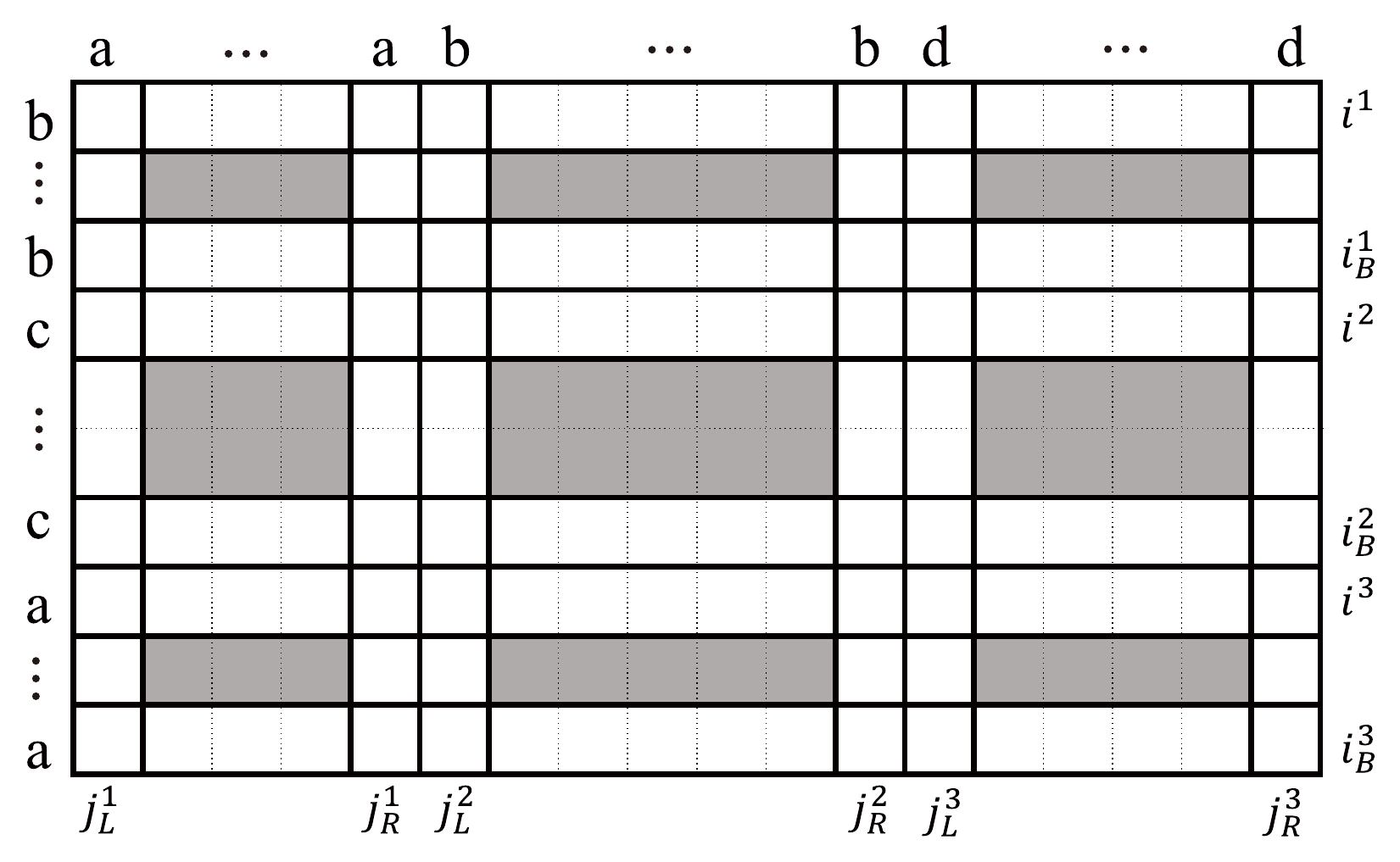}
    \caption{Illustration for $\DS$
    that consists only of the cells of $\DR$ corresponding to
    the maximal run boundaries of $A$ and $B$
    (white rows and columns).
    The gray regions that are surrounded by the box boundaries
    are not stored in $\DS$.}
    \label{fig:DS}
    \end{minipage}

    \end{tabular}
\end{figure} 
Each row and column of $\DS$ is implemented by a linked list as follows:
each cell $\DS[i,j]$ has four links to the upper, lower, left, and right neighbors
in $\DS$ (if these neighbors exist),
plus a diagonal link to the right-lower direction.
This diagonal link from $\DS[i,j]$ points to the first cell $\DS[i+h,j+h]$ 
that is reached by following the right-lower diagonal path from $\DS[i,j]$,
namely, $h \geq 0$ is the smallest integer such that
$i+h = \IBOT$ or $j+h = \JLEFT$.
Clearly $\DS$ occupies $\Theta(mN + nM)$ space.
$\DS$ can answer $\dtw(A,B) = \D[m,n]$ in
$O(m+n)$ time by tracing $O(m+n)$ cells of $\DS$ from $(1,1)$ to $(m,n)$.

For each $1 \leq I < M$ and $1 \leq J < N$,
we consider the region of $\DR$ that is surrounded by the borders of the $I$th and $(I+1)$th runs of $A$,
and the $J$th and $(J+1)$th runs of $B$.
This region is called a \emph{box} for $I,J$, and is denoted by $\B^{I,J}$.
For ease of description, we will sometimes refer to a box $\B^{I,J}$ also in $\D$ and $\DS$.

\subsection{Updating $\DS$ after an edit operation}

Suppose that an edit operation has been performed at position $j^*$ of string $B$
and let $B'$ denote the edited string.
Let $\Dp$ denote the dynamic programming table for $\dtw(A, B')$,
and $\DRp$ the difference representation for $\Dp$.
As Figure~\ref{fig:DR-changed} shows,
the number of changed cells in $\DRp$ can be much smaller than
that of changed cells in $\Dp$ (see also Figure~\ref{fig:worst_case}).
\begin{figure}[th]
    \centering
    \includegraphics[clip,scale=0.45,bb=0 0 800 302]{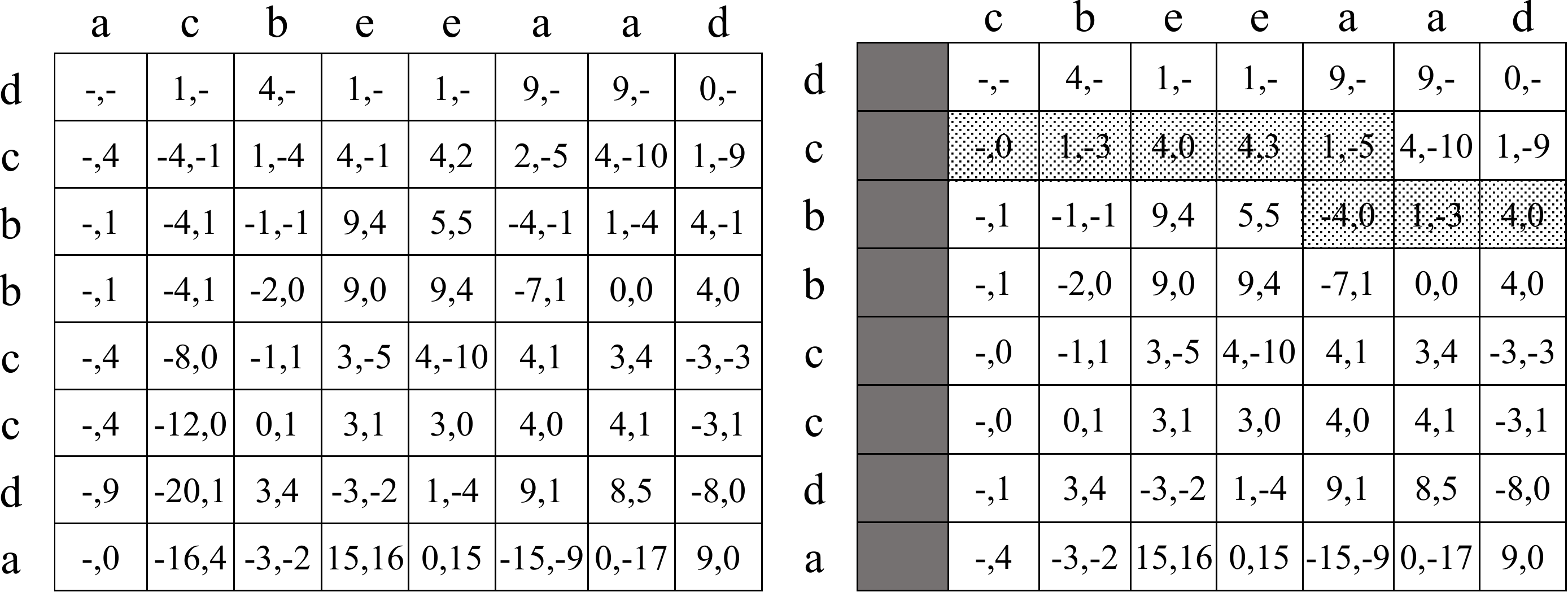}
    \caption{For the running example from Figure~\ref{fig:worst_case}, only the gray cells have different values in the difference representations $\DR$ (left) and $\DRp$ (right).}
    \label{fig:DR-changed}
\end{figure}

Let $\DSp$ denote the sparse table for $\DRp$.
Since $\DS$ consists only of the boundary cells,
the number of changed cells in $\DSp$ can even be much smaller.
In what follows, we show how to efficiently update $\DS$ to $\DSp$.

Because the prefix $B[1..j^*-1]$ remains unchanged
after the edit operation, for any $j < j^*$ we have $\DR[i,j] = \DRp[i,j]$ 
by Lemma~\ref{lem:DR_recursive} and recurrence~(\ref{recurrence}).
Hence, we can restrict ourselves to the indices $j \geq j^*$.
We define $\ell$ as a correcting offset of string indices before and
after the update:
$\ell = -1$ if a character has been inserted at position $j^*$ of $B$,
$\ell = 1$ if a character has been deleted from position $j^*$ of $B$,
and $\ell = 0$ otherwise.
Now, for any $j \geq j^*$,
$B'[j] = B[j+\ell]$ and column $j$ in $\DRp$
corresponds to column $j+\ell$ in $\DR$. 

Let $\B^{I,J}$ be any box on $\DSp$.
For the the top row $\ITOP$ of $\B^{I,J}$,
we use a linked list $\Tlist^{I,J}$ that stores
the column indices $j$~($\JLEFT \leq j \leq \JRIGHT$)
such that $\DS[\ITOP,j+\ell] \neq \DSp[\ITOP,j]$, in increasing order.
We also compute, in each element of the list,
the value for $\Dp[\ITOP,j]$
of the corresponding column index $j$.
We use similar lists $\Blist^{I,J}$, $\Llist^{I,J}$, and $\Rlist^{I,J}$
for the bottom row, left column, and right column
of $\B^{I,J}$, respectively.
We compute these lists when an edit operation is performed to string $B$,
and use them to update $\DS$ to $\DSp$ efficiently.

Let $\changed$ denote the number of cells in our sparse representation
such that $\DS[i+\ell,j]  \neq  \DSp[i,j]$.
In the sequel, we prove:
\begin{theorem} \label{theo:main_theorem}
  Our $\ddtw$ algorithm updates $\DS$ to $\DSp$
  in $O(m \! + \! n \! + \! \changed)$ time.
\end{theorem}

\noindent \textbf{Initial step.}
Suppose that $j^*$ is in the $J$th run of string $B$.
Let $\B^{I,J}$ be any of the $M$ boxes of $\DR$ that contain column $j^*$,
where $\JLEFT \leq j^* \leq \JRIGHT$.
Due to Lemma~\ref{lem:DR_recursive},
$(1,j^*)$ is the only cell in the first row where we may have
$\DSp[1,j^*] \neq \DS[1,j^*+\ell]$.
$\DSp[1,j^*]$ can be easily computed in $O(1)$ time by Lemma~\ref{lem:DR_recursive}.
Then, $\Dp[1,j^*]$ can be computed in $O(j^*) \subseteq O(n)$ time
by tracing the first row and using $\DSp[1,j].L$
for increasing $j = 1, \ldots, j^*$.
The list $\Tlist^{I,J}$ only contains $j^*$ (coupled with $\Dp[1,j^*]$)
if $\DSp[1,j^*] \neq \DS[1,j^*+\ell]$,
and it is empty otherwise.

Editing string $B$ at position $j^*$ incurs some structural changes to $\DS$:
(a) $\B^{I,J}$ gets wider by one (insertion of the same character to a run),
(b) $\B^{I,J}$ gets narrower by one (deletion of a character),
(c) $\B^{I,J}$ is divided into $2M$ or $3M$ boxes (insertion of a different character to a run, or character substitution).

In cases (a) and (b),
the diagonal links of $\B^{I,J}$ need to be updated.
A crucial observation is that the total number of such diagonal links to update
is bounded by $m$ for all the $M$ boxes $\B^{1,J}$, \ldots, $\B^{M,J}$,
since the destinations of such diagonal links are
within the same column of $\DSp$
($\JRIGHT+1$ in case (a), and $\JRIGHT-1$ in case (b)).
For each box $\B^{I,J}$,
if $\JRIGHT-\JLEFT \geq \ITOP-\IBOT$ (i.e. $\B^{I,J}$ is a square or a horizontal rectangle),
then we scan the top row $\ITOP$ from right to left
and fix the diagonal links until encountering
the first cell in $\ITOP$ whose diagonal link needs no updates
(see Figure~\ref{fig:digonal_1}).
The case with $\JRIGHT-\JLEFT < \ITOP-\IBOT$ (i.e. $\B^{I,J}$ is a vertical rectangle) can be treated similarly.
By the above observation,
these costs for all boxes $\B^{I,J}$ that contain the edit position $j^*$
sum up to $O(m)$.

\begin{figure}[tb]
  \centerline{
  \begin{tabular}{c}
    \begin{minipage}{0.48\hsize}
      \centering
        \raisebox{4mm}{\includegraphics[clip,scale=0.7,bb=0 0 250 159]{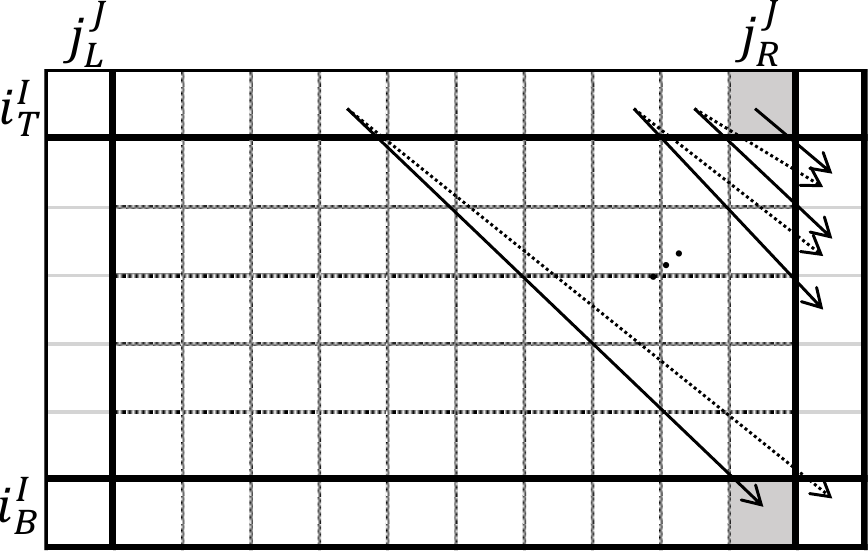}}
        \caption{Case (a) of the initial step.
        The dashed arcs are the old diagonal links in $\DS$,
        and the sold arcs are the modified diagonal links in $\DSp$.
        The gray cells depict cells $(\ITOP, \JRIGHT)$ and $(\IBOT, \JRIGHT)$.}
       \label{fig:digonal_1}
    \end{minipage}
    \begin{minipage}{0.04\hsize}
        \hfill
    \end{minipage}    
    \begin{minipage}{0.48\hsize}
      \centering
        \includegraphics[clip,scale=0.7,bb=0 0 249 162]{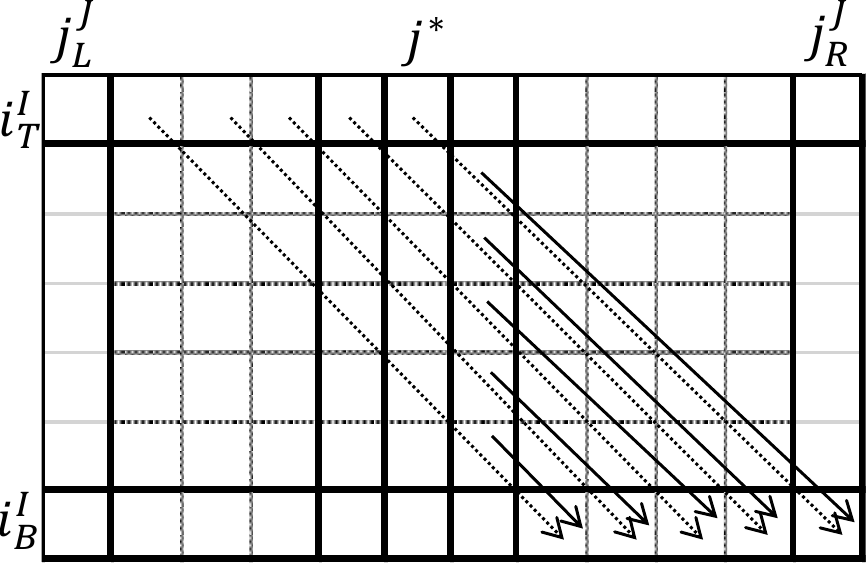}
        \caption{Case (c) of the initial step,
        where character substitution has been performed at position $j^*$.
        The dashed arcs are the old diagonal links in $\DS$
        from row $\ITOP$ up to $j^*$,
        and the sold arcs are the modified diagonal links 
        from new column $j^*$ in $\DSp$.}
        \label{fig:digonal_2}
   \end{minipage}
  \end{tabular}
  }
\end{figure}

In case (a), we shift the right column $\JRIGHT$ of $\DS$
to the right by one position, and reuse it 
as the right column $\JRIGHT+1$ of $\DSp$.
This incurs two new cells $(\ITOP, \JRIGHT)$ and $(\IBOT, \JRIGHT)$
in $\DSp$ (the gray cells in Figure~\ref{fig:digonal_1}).
We can compute $\DSp[\ITOP, \JRIGHT]$
in $O(1)$ time using Lemma~\ref{lem:DR_recursive}.
Now consider to compute $\DSp[i,\JRIGHT+1]$ for the new right column.
Since this right column initially stores $\DS[i,\JRIGHT]$ for the old $\DS$,
using Lemma~\ref{lem:DR_recursive},
we can compute $\DSp[i,\JRIGHT+1]$ in increasing order of $i = 1, \ldots, m$,
from top to bottom, in $O(1)$ time each.
We can compute $\Dp[1,\JRIGHT+1]$ in $O(\JRIGHT)$ time by simply
scanning the first row.
Then, we can compute $\Dp[i,\JRIGHT+1]$ for increasing $i = 2, \ldots, m$,
using $\DSp[i,\JRIGHT+1]$,
and construct $\Rlist^{I,J}$.
This takes a total of $O(\JRIGHT + m) \subseteq O(m+n)$ time.
Finally, $\DSp[\IBOT,\JRIGHT]$ is computed from $\Dp[\IBOT,\JRIGHT+1]$
and $\DSp[\IBOT,\JRIGHT+1].L$ in $O(1)$ time.
Case (b) can be treated similarly.

For case (c),
we consider a sub-case where
a character substitution was performed completely inside a run of $B$,
at position $j^*$.
This divides an existing box $\B^{I,J}$ into three
boxes $\B^{I,J}$, $\B^{I,J+1}$, and $\B^{I,J+2}$.
Thus, there appear three new columns $j^*-1$, $j^*$, and $j^*+1$ in $\DSp$.
Then, the diagonal links for these new columns
can be computed in $O(1)$ time each,
by scanning row $\ITOP$ from $j^*+1$, from right to left
(see Figure~\ref{fig:digonal_2}).
The $\DSp$ values for the cells in these new columns,
as well as the $\Dp$ values for column $j^*+1$, 
can also be computed in similar ways to cases (a) and (b).
The other sub-cases of (c) can be treated similarly.

\vspace*{0.5pc}
\noindent \textbf{Updating cells on row $\ITOP$ and column $\JLEFT$.}
In what follows, suppose that we are given a box $\B^{I,J}$
to the right of the edit position $j^*$,
in which some boundary cell values may have to be updated.
For ease of exposition, we will discuss the simplest case with substitution
where the column indices do not change between $\DS$ and $\DSp$.
The cases with insertion/deletion can be treated similarly
by considering the offset value $\ell$ appropriately.

Now our task is to quickly detect the boundary cells $(i,j)$
of $\B^{I,J}$ such that $\DS[i,j] \neq \DSp[i,j]$,
and to update them.
We assume that the boundary cell values of
the preceding boxes $\B^{I-1,J}$ and $\B^{I,J-1}$ have already been computed.

We consider how to detect
the cells on the top boundary row $\ITOP$ 
and the cells on the left boundary column $\JLEFT$ of box $\B^{I,J}$
that need to be updated, and how to update them.
For this sake, we use the following lemma on the values of $\DR$,
which is immediate from Lemma~\ref{lem:DR_recursive}:
\begin{lemma} \label{lem:DR_propagation}
  Let $1 \leq i \leq m$ and $1 \leq j \leq n$.
  Suppose that for any cell $(i',j')$ with $i' < i$ or $j' < j$,
  the value of $\DRp[i',j']$ has already been computed.
  If $\DR[i,j] \neq \DRp[i,j]$,
  then $\DR[i,j-1].U \neq \DR'[i,j-1].U$ or $\DR[i-1,j].L \neq \DRp[i-1,j].L$. 
\end{lemma}

Intuitively, Lemma~\ref{lem:DR_propagation} states that
the cell $(i,j)$ such that $\DR[i,j] \neq \DRp[i,j]$
must be propagated from its left neighbor or its top neighbor.
We use this lemma for updating the boundaries of each box $\B^{I,J}$
stored in $\DS$.
Recall that the values on the preceding row $\ITOP-1 = i_{\mathrm{B}}^{I-1}$
and on the preceding column $\JLEFT-1 = j_{\mathrm{R}}^{J-1}$
have already been updated.
Then, the cells on $\ITOP$ and $\JLEFT$ of box $\B^{I,J}$
with $\DS[i,j] \neq \DSp[i',j']$ can be found in constant time each,
from the lists $\Blist^{I-1,J}$ and $\Rlist^{I,J-1}$ maintained for
the preceding row $\ITOP - 1 = i_{\mathrm{B}}^{I-1}$ and
preceding column $\JLEFT - 1 = j_{\mathrm{R}}^{J-1}$, respectively.

We process column indices $\Blist^{I-1,J}$ in increasing order,
and suppose that we are currently processing column index $\hat{j} \in \Blist^{I-1,J}$
in the bottom row $i_{\mathrm{B}}^{I-1}$ of the preceding box $\B^{I-1,J}$.
According to the above arguments,
this indicates that the cells $(\ITOP, j)$
in the top row $\ITOP$ of $\B^{I,J}$
that need to be updated (i.e., $\DS[\ITOP,j] \neq \DSp[\ITOP,j]$).
We assume that, for any $j'$ with $\JLEFT \leq j' < \hat{j}$,
the value of $\DSp[\ITOP,j']$ has already been computed. 
Also, we have maintained a partial list for $\Tlist^{I,J}$
where the last element of this partial list stores
the largest $j''$
such that $\JLEFT \leq j'' < \hat{j}$ and
$\DS[\ITOP, j''] \neq \DSp[\ITOP, j'']$,
together with the value of $\D'[\ITOP,j'']$.
Now it follows from Lemma~\ref{lem:DR_recursive} that
both $\DSp[\ITOP, \hat{j}].U$ and $\DSp[\ITOP, \hat{j}].L$ can be
respectively computed in constant time from
$\DSp[\ITOP-1,\hat{j}].L$ and $\DSp[\ITOP,\hat{j}-1].U$,
and thus we can check whether $\DS[\ITOP,\hat{j}] \neq \DSp[\ITOP,\hat{j}]$
in constant time as well.
In case $\DS[\ITOP,\hat{j}] \neq \DSp[\ITOP,\hat{j}]$,
we append $\hat{j}$ to the partial list for $\Tlist^{I,J}$.
By the definition of $\DS$,
we have $\Dp[\ITOP,\hat{j}] = \Dp[\ITOP-1,\hat{j}] - \DSp[\ITOP, \hat{j}].U$.
Since $\Dp[\ITOP-1,\hat{j}] = \Dp[i_{\mathrm{B}}^{I-1},\hat{j}]$ is stored
with the current column index $\hat{j}$ in the list $\Blist^{I-1,J}$,
$\Dp[\ITOP,\hat{j}]$ can also be computed in constant time.

Suppose we have processed cell $(\ITOP,\hat{j})$.
We perform the same procedure as above
for the right-neighbor cells $(\ITOP,\hat{j}+p)$
with $p = 1$ and increasing $p$,
until encountering the first cell $(\ITOP,\hat{j}+p)$
such that (1) $\DS[\ITOP,\hat{j}+p] = \DSp[\ITOP,\hat{j}+p]$,
(2) $\hat{j}+p \in \Blist^{I-1,J}$,
or (3) $\hat{j}+p = \JRIGHT+1$.
In cases (1) and (2), we move on to the next element of in $\Blist^{I-1,J}$,
and perform the same procedure as above.
We are done when we encounter case (3) or $\Blist^{I-1,J}$ becomes empty.
The total number of cells $(\ITOP,\hat{j}+p)$ 
for all boxes in $\DSp$ is bounded by $\changed$.

In a similar way, 
we process row indices $\Rlist^{I,J-1}$ in increasing order,
update the cells on the left column $\JLEFT$,
and maintain another partial list for $\Llist^{I,J}$.

\vspace*{0.5pc}
\noindent \textbf{Updating cells on row $\IBOT$ and column $\JRIGHT$.}
Let us consider how to detect
the cells on the bottom row $\IBOT$ 
and the cells on the right column $\JRIGHT$ of box $\B^{I,J}$
that need to be updated, and how to update them.

The next lemma shows
monotonicity on the values of $\D$ inside each $\B^{I,J}$.
\begin{lemma}[\cite{FroeseJRW2020}] \label{lem:monotone_D}
  For any $(i,j)$ with $1 \leq i \leq m$ and $\JLEFT < j \leq \JRIGHT$,
  $\D[i,j] \geq \D[i,j-1]$.
  For any $(i,j)$ with $\ITOP < i \leq \IBOT$ and $1 \leq j \leq n$,
  $\D[i,j] \geq \D[i-1,j]$.
\end{lemma}
The next corollary is immediate from Lemma~\ref{lem:monotone_D}.
\begin{corollary} \label{coro:non-negative_DR}
  For any cell $(i,j)$ with $1 \leq i \leq m$ and $\JLEFT < j \leq \JRIGHT$,
  $\DR[i,j].L \geq 0$.
  Also, for any cell $(i,j)$ with $\ITOP < i \leq \IBOT$ and $1 \leq j \leq n$,
  $\DR[i,j].U \geq 0$.
\end{corollary}

Now we obtain the next lemma, which is a key to our algorithm.
\begin{lemma} \label{lem:key_lemma_diagonal}
  For any cell $(i,j)$ with
  $\ITOP+1 < i \leq \IBOT$ and $\JLEFT+1 < j \leq \JRIGHT$,
  $\DR[i,j] = \DR[i-1,j-1]$.
\end{lemma}
\begin{proof}
  By Corollary~\ref{coro:non-negative_DR},
  $\DR[i-1,j].L \geq 0$ and $\DR[i,j-1].U \geq 0$
  for $\ITOP+1 < i \leq \IBOT$ and $\JLEFT+1 < j \leq \JRIGHT$.
  Thus clearly $\min\{\DR[i-1,j].L, \DR[i,j-1].U, 0\} = 0$.
  Therefore, for the value of $z$ in Lemma~\ref{lem:DR_recursive},
  we have $z = (a_i - b_j)^2$, which leads to
  \begin{eqnarray}
    \DR[i,j].U & = & (a_i-b_j)^2 - \DR[i-1,j].L \label{eqn:U} \\
    \DR[i,j].L & = & (a_i-b_j)^2 - \DR[i,j-1].U \label{eqn:L} 
  \end{eqnarray}
  By applying equation~(\ref{eqn:L}) to the $\DR[i-1,j].L$ term of equation~(\ref{eqn:U}),
  we get
  \[ \DR[i,j].U = (a_i-b_j)^2 - ((a_{i-1}-b_j)^2 - \DR[i-1,j-1].U). \]
  Recall that $a_i=a_{i-1}$,
  since we are considering cells in the same box $\B^{I,J}$.
  Thus $\DR[i,j].U = \DR[i-1,j-1].U$.
  By applying equation~(\ref{eqn:U}) to the $\DR[i,j-1].U$ term of equation~(\ref{eqn:L}), we similarly obtain $\DR[i,j].L = \DR[i-1,j-1].L$.
  \qed
\end{proof}

For any $\ITOP+1 < i \leq \IBOT$ and $\JLEFT+1 < j \leq \JRIGHT$,
let $\ell$ be the smallest positive integer that satisfies 
$i-\ell = \ITOP+1$ or $j-\ell = \JLEFT+1$.
By Lemma~\ref{lem:key_lemma_diagonal},
for any cell $(i,j)$ on the bottom row $\IBOT$ or on the right column $\JRIGHT$,
we have $\DS[i,j] = \DR[i-\ell,j-\ell]$ and $\DSp[i,j] = \DRp[i-\ell,j-\ell]$.
This means that $\DS[i,j] \neq \DSp[i,j]$ iff
$\DR[i-\ell,j-\ell] \neq \DRp[i-\ell,j-\ell]$.
Thus, finding cells $(i,j)$ with $\DS[i,j] \neq \DSp[i,j]$ on
the bottom row $\IBOT$ or on the right column $\JRIGHT$
reduces to finding cells $(i',j')$ with $\DR[i',j'] \neq \DRp[i',j']$
on the row $\ITOP+1$ or on the column $\JLEFT+1$.
See Figure~\ref{fig:diagonal_propagation}.

\begin{figure}[tb]
  \centerline{
  \begin{tabular}{c}
    \begin{minipage}{0.49\hsize}
      \centering
         \includegraphics[clip,scale=0.7,bb=0 0 306 170]{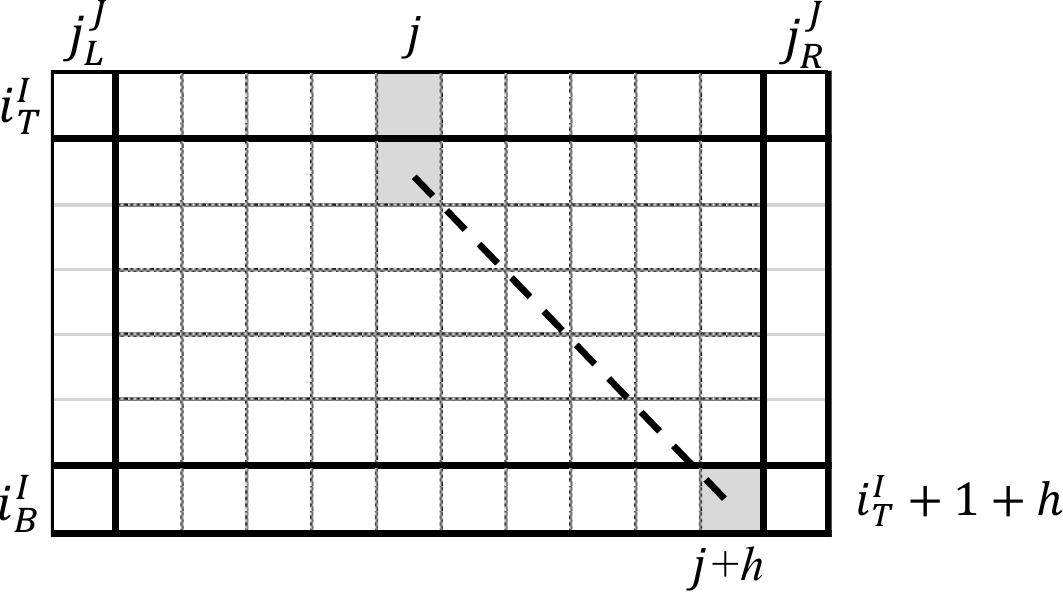}
      \caption{Diagonal propagation of $\DR[i,j] \neq \DRp[i,j]$ inside box $\B^{I,J}$.}
    \label{fig:diagonal_propagation}
    \end{minipage}
    \begin{minipage}{0.01\hsize}
        \hfill
    \end{minipage}    
    \begin{minipage}{0.49\hsize}
      \centering
         \includegraphics[clip,scale=0.7,bb=0 0 263 180]{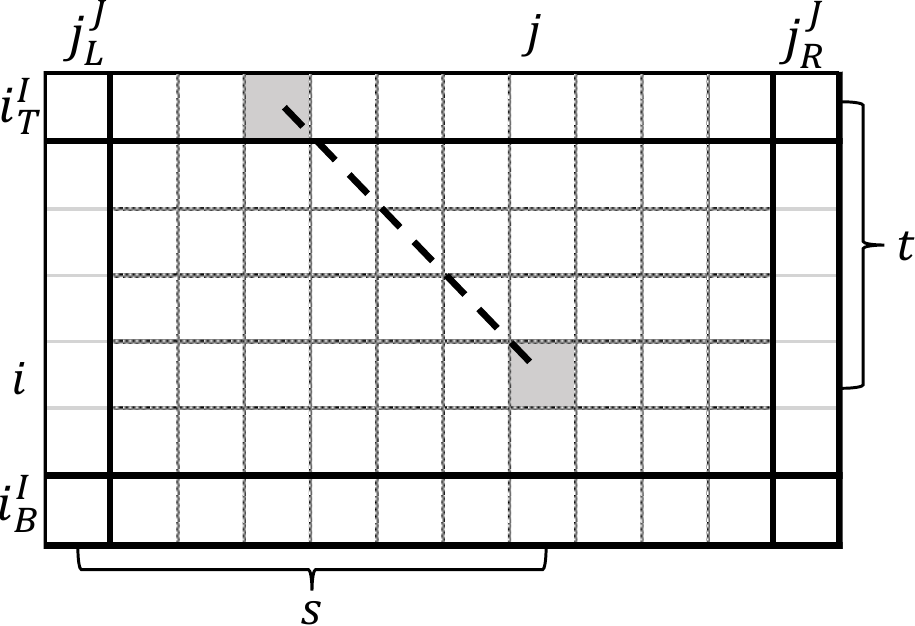}
      \caption{Illustration for the case where $s > t$ in Lemma~\ref{lem:s-t_t-s}.}    \label{fig:s-t_t-s}
   \end{minipage}
  \end{tabular}
  }
\end{figure} 

We have shown how to compute $\Tlist^{I,J}$ for the top row $\ITOP$
and $\Llist^{I,J}$ for the left column $\JLEFT$.
We here explain how to use $\Tlist^{I,J}$
(we can use $\Llist^{I,J}$ in a symmetric manner).
We process column indices in $\Tlist^{I,J}$ in increasing order,
and suppose that we are currently processing column index $\hat{j} \in \Tlist^{I,J}$
in the top row $\ITOP$ of the current box $\B^{I,J}$.
We check whether $\DR[\ITOP+1,\hat{j}] \neq \DRp[\ITOP+1,\hat{j}]$.
For this sake, we need to know the values of $\DR[\ITOP+1,\hat{j}]$ and $\DRp[\ITOP+1,\hat{j}]$.
Recall that, by Lemma~\ref{lem:key_lemma_diagonal},
$\DR[\ITOP+1,\hat{j}]$ is equal to $\DR[\ITOP+1+h, \hat{j}+h]$~($= \DS[\ITOP+1+h, \hat{j}+h]$)
on the bottom row $\IBOT$ (if $\ITOP+1+h = \IBOT$)
or on the right column $\JRIGHT$ (if $\hat{j}+h = \JRIGHT$),
where $h > 0$.
Since the cell $(\ITOP+1+h, \hat{j}+h)$ can be retrieved in constant time
by the diagonal link from the cell $(\ITOP,\hat{j}-1)$ on the
top row $\ITOP$,
we can compute $\DR[\ITOP+1,\hat{j}]$ in constant time,
applying Lemma~\ref{lem:key_lemma_diagonal} to the upper-left direction.

Computing $\DRp[\ITOP+1,\hat{j}]$ is more involved.
By Lemma~\ref{lem:DR_recursive},
we can compute $\DRp[\ITOP+1,\hat{j}]$ from $\DRp[\ITOP,\hat{j}].L$ and
$\DRp[\ITOP+1,\hat{j}-1].U$.
Since $(\ITOP,\hat{j})$ is on the top row $\ITOP$,
$\DRp[\ITOP,\hat{j}].L = \DSp[\ITOP,\hat{j}].L$ has already been computed.
Consider to compute $\DRp[\ITOP+1,\hat{j}-1].U$.
Since $\DRp[\ITOP+1,\hat{j}-1].U = \Dp[\ITOP+1,\hat{j}-1] - \Dp[\ITOP,\hat{j}-1]$,
it suffices to compute $\Dp[\ITOP,\hat{j}-1]$ and $\Dp[\ITOP+1,\hat{j}-1]$.
By definition, $\Dp[\ITOP,\hat{j}-1] = \Dp[\ITOP,\hat{j}] - \DRp[\ITOP,\hat{j}].L$.
Since $\hat{j} \in \Tlist^{I,J}$, we can retrieve the value of $\Dp[\ITOP,\hat{j}]$
from the current element of the list $\Tlist^{I,J}$, in $O(1)$ time.
Since $\DRp[\ITOP,\hat{j}].L = \DSp[\ITOP,\hat{j}].L$,
we can compute $\Dp[\ITOP,\hat{j}-1]$ in $O(1)$ time.

What remains is how to compute $\Dp[\ITOP+1,\hat{j}-1]$.
We use the next lemma.

\begin{lemma} \label{lem:s-t_t-s}
  For any cell $(i,j)$ with
  $\ITOP+1 < i \leq \IBOT$ and $\JLEFT+1 < j \leq \JRIGHT$,
  let $s = j - \JLEFT$ and $t = i - \ITOP$.
  Then,
  \[
   \D[i,j] = \D[\ITOP+\max\{t-s,0\},\JLEFT + \max\{s-t,0\}] + \min\{s,t\} \cdot (a_i-b_j)^2.
  \]
\end{lemma}

\begin{proof}
  Consider the case where $s > t$.
  By applying Lemma~\ref{lem:monotone_D} to recurrence~(\ref{recurrence}),
  we obtain $\D[i,j] = \D[i-1,j-1] + (a_i-b_j)^2$.
  Since $a_i = a_{i'}$ and $b_j = b_{j'}$ for $\ITOP < i' < i$
  and $\JLEFT < j' < j$,
  by repeatedly applying Lemma~\ref{lem:monotone_D} to the above equation,
  we get $\D[i,j] = \D[\ITOP,\JLEFT+(s-t)] + t \cdot(a_i-b_j)^2$.
  See also Figure~\ref{fig:s-t_t-s}.
  The case $s \leq t$ is similar and
  we obtain $\D[i,j] = \D[\ITOP+(t-s),\JLEFT] + s \cdot (a_i-b_j)^2$.
  By merging the two equations for $s > t$ and $s \leq t$,
  we obtain the desired equation.
  \qed
\end{proof}

Let $k = \hat{j} - \JLEFT$.
Since $\JLEFT+1 < \hat{j}$, $k \geq 2$.
Since $s=\hat{j}-1-\JLEFT=k-1$, $t=\ITOP+1-\ITOP=1$, and $k \geq 2$,
we get $s \geq t$.
Thus it follows from Lemma~\ref{lem:s-t_t-s}
that
$$\Dp[\ITOP+1,\hat{j}-1] \! = \! \Dp[\ITOP,\JLEFT+(k-2)] + (A[\ITOP]-B[\hat{j}])^2 \! = \! \Dp[\ITOP,\hat{j}-2]+(A[\ITOP]-B[\hat{j}])^2.$$
Since the value $\Dp[\ITOP,\hat{j}]$ is already computed
and stored in the corresponding element of $\Delta_T^{I,J}$,
we can compute, in $O(1)$ time, $\Dp[\ITOP,\hat{j}-2]$ by
\begin{eqnarray*}
  \Dp[\ITOP,\hat{j}-2] & = & \Dp[\ITOP,\hat{j}] - \DRp[\ITOP,\hat{j}].L-\DRp[\ITOP,\hat{j}-1].L \\
  & = & \Dp[\ITOP,\hat{j}] - \DSp[\ITOP,\hat{j}].L-\DSp[\ITOP,\hat{j}-1].L.
\end{eqnarray*}
Thus, we can determine in $O(1)$ time
whether $\DR[\ITOP+1,\hat{j}] \neq \DRp[\ITOP+1,\hat{j}]$,
and hence whether $\DS[\ITOP+1+h,\hat{j}+h] \neq \DSp[\ITOP+1+h,\hat{j}+h]$.

Suppose $\DS[\ITOP+1+h,\hat{j}+h] \neq \DSp[\ITOP+1+h,\hat{j}+h]$.
Then we need to compute $\Dp[\ITOP+1+h,\hat{j}+h]$.
This can be computed in constant time using Lemma~\ref{lem:s-t_t-s},
by $\Dp[\ITOP+1+h,\hat{j}+h]=\Dp[\ITOP,\hat{j}-1]+(h+1) \cdot (A[\ITOP]-B[\hat{j}])^2$, where $\Dp[\ITOP,\hat{j}-1]=\Dp[\ITOP,\hat{j}]-\DRp[\ITOP,\hat{j}].L$.
We add the column index $\hat{j}+h$ to list $\Blist^{I,J}$ if $\ITOP+1+h = \IBOT$,
and/or add the row index $\ITOP+1+h$ to list $\Rlist^{I,J}$ if $\hat{j}+h = \JRIGHT$,
together with the value of $\Dp[\ITOP+1+h,\hat{j}+h]$.

\begin{figure}[b!]
    \centering
    \includegraphics[clip,width=0.99\linewidth,bb=0 0 698 378]{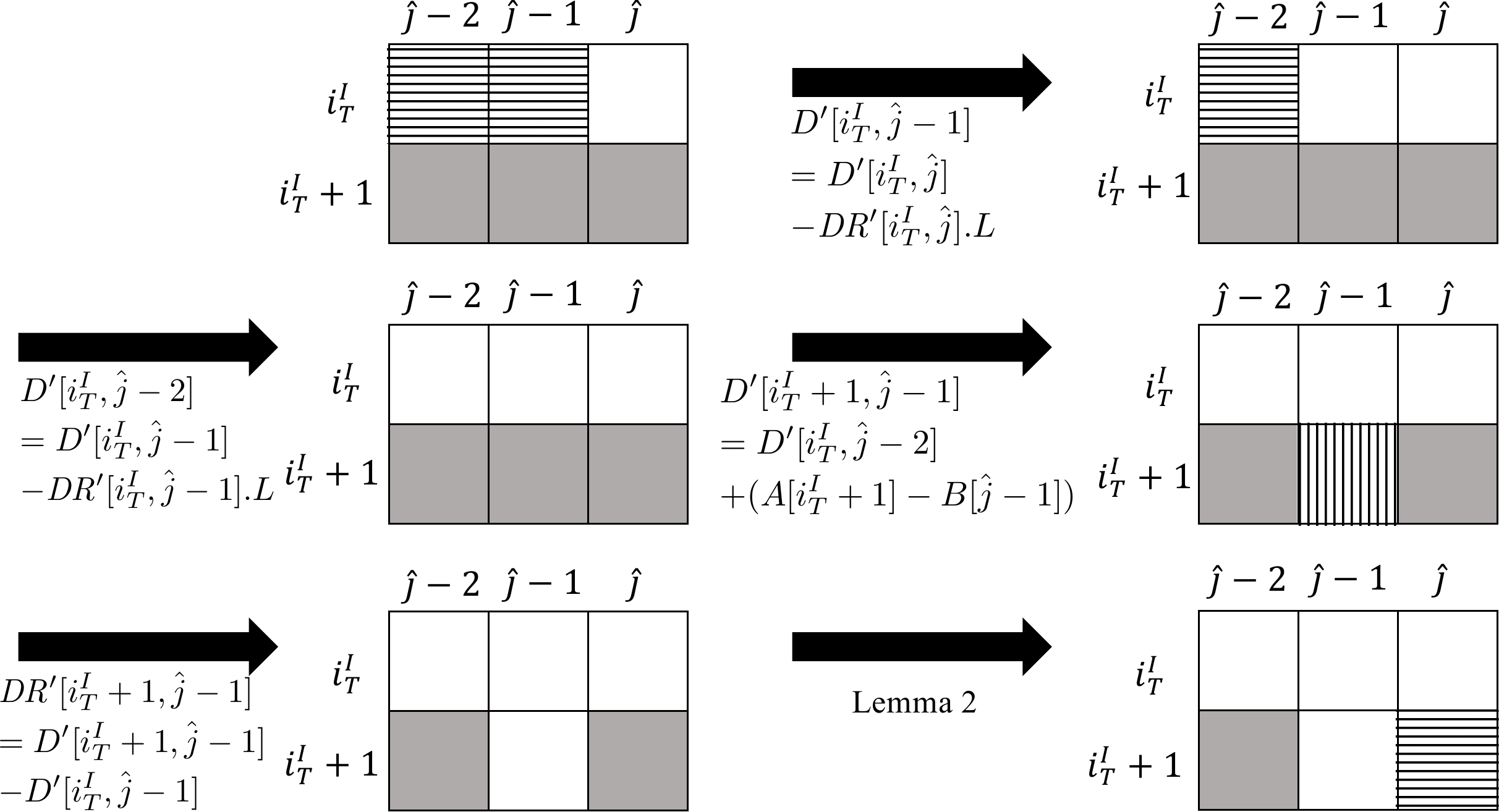}
    \caption{Illustration for the process of computing $\DRp[\ITOP+1,\hat{j}]$. The gray cells show those for which both values of $\Dp$ and $\DRp$ are unknown, the vertically striped cells show those for which only the value of $\Dp$ is known, the horizontally striped cells show those for which only the value of $\DRp$ is known, and the white cells show those for which both values of $\Dp$ and $\DRp$ are known. At the final step (lower-right), the desired value $\DRp[\ITOP+1,\hat{j}]$ has been computed.}
    \label{fig:computation_flow}
\end{figure}

The above process of computing $\DRp[\ITOP+1,\hat{j}]$
is illustrated in Figure~\ref{fig:computation_flow}.
Suppose we have processed cell $(\ITOP+1,\hat{j})$.
We perform the same procedure as above
for the right-neighbor cells $(\ITOP+1,\hat{j}+q)$
with $q = 1$ and increasing $q$,
until encountering the first cell $(\ITOP+1,\hat{j}+q)$
such that (1) $\DR[\ITOP+1,\hat{j}+q] = \DRp[\ITOP+1,\hat{j}+q]$,
(2) $\hat{j}+q \in \Tlist^{I,J}$,
or (3) $\hat{j}+q = \JRIGHT+1$.
In cases (1) and (2), we remove $\hat{j}$ from $\Tlist^{I,J}$
and move to the next element of in $\Tlist^{I,J}$.
We are done when we encounter case (3) or $\Tlist^{I,J}$ becomes empty.
By Lemma~\ref{lem:key_lemma_diagonal},
the total number of cells $(\ITOP+1,\hat{j}+q)$
for all boxes in $\DSp$ is $O(\changed)$.

\vspace*{0.5pc}
\noindent \textbf{Batched updates.}
Our algorithm can efficiently support \emph{batched updates}
for insertion, deletion, substitution of a run of characters.
\begin{theorem} \label{theo:batched_update}
  Let $\Bp$ be the string after
  a run-wise edit operation on $B$, and let $n' = |\Bp|$.
  $\DS$ can be updated to $\DSp$ in $O(m+\max\{n,n'\}+\changed')$ time
  where $\changed'$ denotes the number of cells where
  the values differ between $\DS$ and $\DSp$.
\end{theorem}

Since $n'$ is the length of the string $|\Bp|$ after modification,
$\changed'$ in Theorem~\ref{theo:batched_update}
is bounded by $O(mN+\max\{n',n\}M)$.
Thus, we can perform a batched run-wise update on our sparse table $\DS$ in 
worst-case $O(m+\max\{n,n'\}+\changed') \subseteq O(mN + \max\{n,n'\}M)$ time.
Let $k$ be the total number of characters
that are involved in a run-wise batched edit operation from $B$ to $\Bp$
(namely, a run of $k$ characters is inserted,
a run of $k$ characters is deleted,
or a run of $k_1$ characters is substituted for a run of $k_2$ characters
with $k = k_1 + k_2$).
Then a na\"ive $k$-time applications of Theorem~\ref{theo:main_theorem}
to the run-wise batched edit operation
requires $O(k(m+n+\changed)) \subseteq O(k(mN+nM))$ time.
Since $n' \leq n + k$,
the batched update of Theorem~\ref{theo:batched_update}
is faster than the na\"ive method by a factor of $k$
whenever $k \in O(n)$.
We also remark that our batched update algorithm
is at least as efficient as
building the sparse DP table of Froese et al.'s algorithm~\cite{FroeseJRW2020}
from scratch using $\Theta(mN+\max\{n,n'\}M)$ time and space.

%% file: evaluation.tex
\subsection{Evaluation of $\changed$} \label{sec:evaluation}

As was proven previously,
our $\ddtw$ algorithm works in $O(m+n+\changed)$ time
per edit operation on one of the strings.
In this subsection, we analyze how large the $\changed$ would be
in theory and practice.
Although $\changed = \Theta(mN + nM)$ in the worst case
for some strings (Theorem~\ref{theo:worst_case_changed}),
our preliminary experiments shown below suggest
that $\changed$ can be much smaller than $mN + nM$ in many cases.

\begin{theorem} \label{theo:worst_case_changed}
Consider strings $A = A_1^{k} \cdots A_M^{k}$ and $B = B_1^{l} \cdots B_N^{l}$
of RLE sizes $M$ and $N$, respectively,
where $|A| = m = kM$ and $|B| = n = lN$.
We assume lexicographical orders of characters
as $A_{I-1} > A_I$ for $1 < I \leq M$,
$B_{J-1} < B_J$ for $1 < J \leq N$, and $A_M > B_N$.
If we delete $B[1]$ from $B$, then $\changed = \Omega(mN + nM)$.
\end{theorem}

We have also conducted preliminary experiments to estimate
practical values of $\changed$, using randomly generated strings.
For simplicity, we set $m = n$ and $M = N$ for all experiments.
We fixed the alphabet size $|\Sigma| = 26$ throughout our experiments.
\begin{figure}[h!]
    \centering
    \includegraphics[width=0.7\linewidth,bb=0 0 414 224]{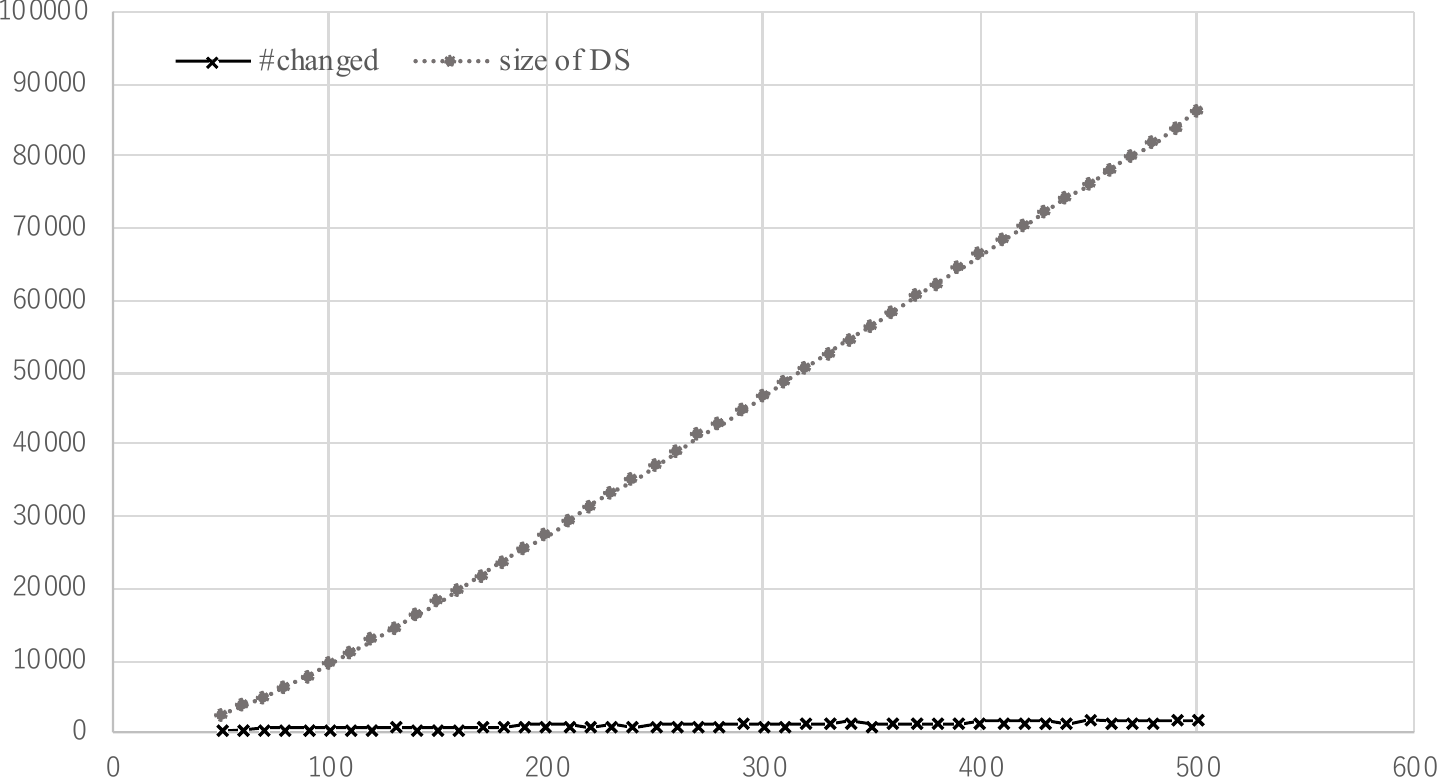}

    \vspace{1pc}
    
    \includegraphics[width=0.7\linewidth,bb=0 0 346 179]{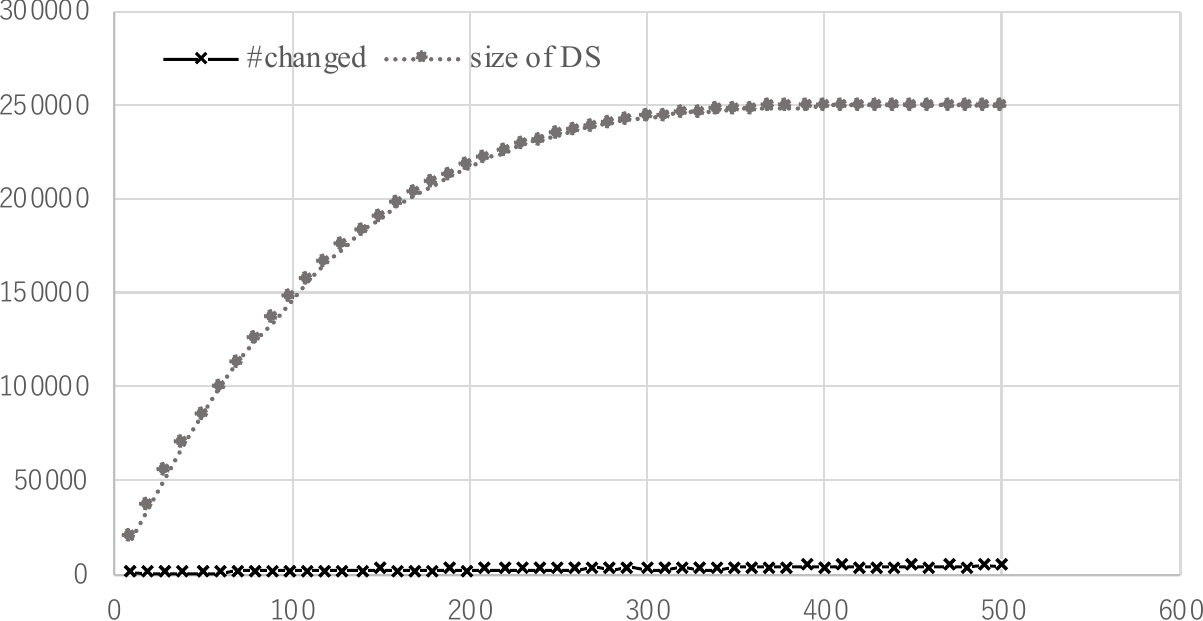}
    \caption{Comparisons of the values of $\changed$ and the sizes 
    of the sparse table $\DS$ on two randomly generated strings $A$ and $B$. Upper: With fixed RLE size $N = M = 50$ and varying lengths $n = m$ from $50$ to $500$ (horizontal axis). Lower: With fixed length $n = m = 500$ and varying RLE sizes $N = M$ from $10$ to $500$ (horizontal axis).}
    \label{fig:experiments}
\end{figure}
In the first experiment, we fixed the RLE size $M = N = 50$,
randomly generated two strings $A$ and $B$ of varying lengths $m = n$
from $50$ to $500$,
and compared the values of $\changed$ and the sizes of $\DS$.
For each $m$,
we randomly generated 50 pairs of strings $A$ and $B$ of length $m$ each,
and took the average values for $\changed$ and the sizes of $\DS$
when $B[1]$ was deleted from $B$.
In the second experiment, we fixed the string length
$m = n = 500$ and randomly generated two strings $A$ and $B$ of varying
RLE sizes $M = N$ from $10$ to $500$.
For each $M$,
we randomly generated 50 pairs of strings $A$ and $B$ of RLE size $M$,
and took the average values for $\changed$ and the sizes of $\DS$
when $B[1]$ was deleted from $B$.
The results are shown in Figure~\ref{fig:experiments}.
In both experiments, $\changed$ is much smaller than the size of $\DS$.
It is noteworthy that even when the values of $M$~($= N$) and $m$~($= n$)
are close, the value of $\changed$ stayed very small.
This suggests that our algorithm can be fast also on
strings that are \emph{not} RLE-compressible.

%% file: ack.tex
\section*{Acknowledgments}

This work was supported by JSPS KAKENHI Grant Numbers
JP18K18002 (YN),
JP17H01697 (SI),
JP20H04141 (HB), JP18H04098 (MT),
and JST PRESTO Grant Number JPMJPR1922 (SI).

%% file: appendix.tex
\section{Appendix: Omitted proofs} \label{sec:omitted_proofs}

\subsection{Proof for Theorem~\ref{theo:worst_case_changed}}

To prove Theorem~\ref{theo:worst_case_changed},
we establish the following lemma.

\begin{lemma} \label{lem:path_min}
  Let $\D$ be the dynamic programming table for the above strings $A$ and $B$.
  Then,
  \[
  \min\{\D[i-1,j], \D[i,j-1], \D[i-1,j-1]\} = 
  \begin{cases}
    \D[i,j-1] & \mbox{if } i < j, \\
    \D[i-1,j] & \mbox{if } i > j, \\
    \D[i-1,j-1] & \mbox{if } i = j. \\
  \end{cases}
  \]
\end{lemma}

\begin{proof}
  By recurrence~(\ref{recurrence}),
  the argument holds for any cells $(1,j)$ and $(i,1)$,
  where $1 \leq i \leq m$ and $1 \leq j \leq n$.
  For any cell $(i,j)$ with $i > 1$ or $j > 1$,
  suppose that the argument holds for any $(i', j')$ with $i' < i$ and $j' < j$.
  We consider the five following cases:
  \begin{enumerate}
  \item \label{case:i<j-1}
    Case $i < j-1$:
    For the cell $(i-1,j)$,
    it follows from the inductive hypothesis and recurrence~(\ref{recurrence})
    that $\D[i-1,j] = \D[i-1,j-1] + (a_{i-1}-b_{j})^2$.
    Since $(a_{i-1}-b_{j})^2 \geq 0$, $\D[i-1,j] \geq \D[i-1,j-1]$.
    Similarly, for the cells $(i-1,j-1)$ and $(i,j-1)$,
    $\D[i-1,j-1] = \D[i-1,j-2] + (a_{i-1}-b_{j-1})^2$ and
    $\D[i,j-1] = \D[i,j-2]+(a_{i}-b_{j-1})^2$.
    Since $a_{i-1} \geq a_{i}$, $a_{i-1}-b_{j-1} \geq a_{i}-b_{j-1}$.
    Because $a_{i}-b_{j-1}>0$, 
    it holds that $(a_{i-1}-b_{j-1})^2 \geq (a_{i}-b_{j-1})^2$.
    By the inductive hypothesis, $\D[i-1,j-2] \geq \D[i,j-2]$.
    Thus we have $\D[i-1,j-2]+(a_{i-1}-b_{j-1})^2 \geq \D[i,j-2]+(a_{i}-b_{j-1})^2$,
    which implies $\D[i-1,j-1] \geq \D[i,j-1]$.

  \item \label{case:i=j-1}
    Case $i=j-1$:
    Analogously to Case~\ref{case:i<j-1},
    we get $D[i-1,j] \geq D[i-1,j-1]$.
    For the cells $(i-1,j-1)$ and $(i,j-1)$,
    it follows from the inductive hypothesis and recurrence~(\ref{recurrence})
    that $\D[i-1,j-1] = \D[i-1,j-2] + (a_{i-1}-b_{j-1})^2$
    and $\D[i,j-1] = \D[i-1,j-2] + (a_{i}-b_{j-1})^2$.
    Since $(a_{i-1}-b_{j-1})^2 \geq (a_{i}-b_{j-1})^2$,
    $\D[i-1,j-2]+(a_{i-1}-b_{j-1})^2 \geq \D[i-1,j-2]+(a_{i}-b_{j-1})^2$,
    which implies $\D[i-1,j-1] \geq \D[i,j-1]$.

  \item Case $i-1>j$: By symmetric arguments to Case~\ref{case:i<j-1}.

  \item Case $i-1=j$: By symmetric arguments to Case~\ref{case:i=j-1}.
    
  \item Case $i=j$:
    For the cells $(i-1,j)$ and $(i,j-1)$,
    by the inductive hypothesis 
    $\min\{\D[i-2,j],\D[i-1,j-1],\D[i-1,j]\}=D[i-1,j-1]$ and
    $\min\{\D[i-1,j-1],\D[i,j-2],\D[i,j-1]\}=D[i-1,j-1]$.
    Thus $\D[i-1,j] \geq \D[i-1,j-1]$ and $\D[i,j-1] \geq D[i-1,j-1]$.
  \end{enumerate}
  \qed
\end{proof}

We are ready to prove Theorem~\ref{theo:worst_case_changed}.

\begin{proof}
  For simplicity, we assume that the column indices of $\D$ begin with 0.
  In the grid graph $\mathcal{G}_{m,n}$ over $\D$,
  we assign a weight $(a_i - b_j)^2$ to each in-coming edge of cell $(i,j)$.
  We also consider the grid graph $\mathcal{G}_{m-1,n}$ over $\Dp$
  obtained by removing $(1,0), \ldots, (m,0)$ from $\mathcal{G}_{m,n}$.

  Consider a cell $(i,j)$ with $1<i<j$
  that is on the top row of some box $\B^{I,J}$,
  namely $i = \ITOP = Ik+1$ for some $1 \leq I \leq M-1$.
  By Lemma~\ref{lem:path_min},
  $p_1=(1,0),\dots,(i,i-1),\dots,(i,j)$ is the minimum weight path
  from $(1,0)$ to $(i,j)$.
  Similarly, $p'_1=(1,1),\dots,(i,i),\dots,(i,j)$
  is the minimum weight path from $(1,1)$ to $(i,j)$.

  For the cell $(i-1,j)$ that is the upper neighbor of $(i,j)$
  and is on the bottom row of box $\B^{I-1,J}$,
  by analogous arguments to the above,
  $p_2=(1,0),\dots,(i-1,i-2),\dots,(i-1,j)$ is the minimum weight path
  from $(1,0)$ to $(i-1,j)$,
  and $p'_2=(1,1),\dots,(i-1,i-1),\dots,(i-1,j)$ is the minimum weight path
  from $(1,1)$ to $(i-1,j)$.

  Let $p_3$ be the sub-path of $p_1$ and $p_2$ ending at $(i-1,i-2)$,
  $p_4$ the sub-path of $p'_1$ and $p'_2$ ending at $(i-1,i-1)$,
  $p_5$ the sub-path of $p_1$ and $p'_1$ from $(i,i)$ to $(i,j)$,
  and $p_6$ the sub-path of $p_2$ and $p'_2$ from $(i-1,i-1)$ to $(i-1,j)$.
  Let $e_1$ be the edge from $(i-1,i-2)$ to $(i,i-1)$,
  $e_2$ be the edge from $(i,i-1)$ to $(i,i)$,
  $e_3$ be the edge from $(i-1,i-2)$ to $(i-1,i-1)$,
  and $e_4$ be the edge from $(i-1,i-1)$ to $(i,i)$.
  See Figure~\ref{fig:lowerbound_analysis} that depicts these paths and edges.
  \begin{figure}[h]
    \centering
    \includegraphics[clip,width=70mm,bb=0 0 278 173]{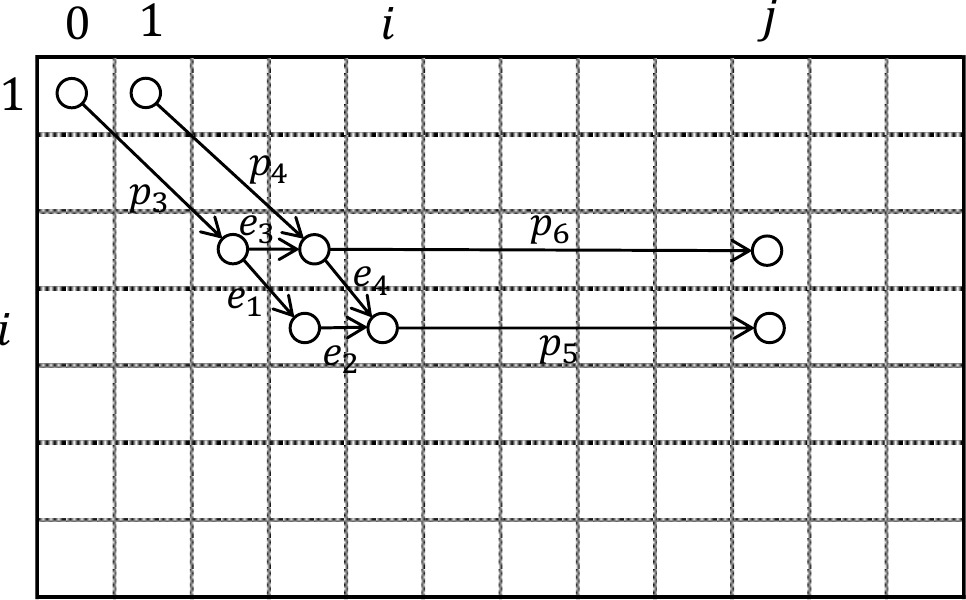}
    \caption{Minimum weight paths to $(i,j)$ and $(i-1,j)$ over $\D$ and $\Dp$.}
    \label{fig:lowerbound_analysis}
  \end{figure}

  For any path $p$ in $\mathcal{G}_{m,n}$,
  let $\cost(p)$ denote the total weights of edges in $p$.
  Now we have 
  $\cost(p_1)=\cost(p_3)+\cost(e_1)+\cost(e_2)+\cost(p_5)$ and 
  $\cost(p'_1)=\cost(p_4)+\cost(e_4)+\cost(p_5)$.
  By the definition of DTW,
  $\D[i,j]$ stores the cost of the minimum weight path
  from $(1,0)$ to $(i,j)$,
  and $\Dp[i,j]$ stores the cost of the minimum weight path
  from $(1,1)$ to $(i,j)$.
  Thus $\D[i,j] = \cost(p_3)+\cost(e_1)+\cost(e_2)+\cost(p_5)$
  and $\Dp[i,j] = \cost(p_4)+\cost(e_4)+\cost(p_5)$.
  Similarly, $\D[i-1,j]=\cost(p_3)+\cost(e_3)+\cost(p_6)$ and
  $\Dp[i-1,j]=\cost(p_4)+\cost(p_6)$.

  Now $\DS[i,j].U = \D[i-1,j]-\D[i,j]=\cost(e_3)+\cost(p_6)-(\cost(e_1)+\cost(e_2)+\cost(p_5))$,
  and $\DSp[i,j].U= \Dp[i-1,j]-\Dp[i,j]=\cost(p_6)-(\cost(e_4)+\cost(p_5))$.
  This leads to $\DS[i,j].U-\DSp[i,j].U=\cost(e_3)+\cost(e_4)-(\cost(e_1)+\cost(e_2))$.
  Recall that $\cost(e_2)=\cost(e_4) = (A[i] - B[j])^2$.
  Thus $\DS[i,j].U-\DSp[i,j].U=\cost(e_3)-\cost(e_1)$.
  Also, because $A[i] \neq A[i-1]$,
  $(A[i]-B[i-1])^2 = \cost(e_1) \neq \cost(e_3) = (A[i-1]-B[i-1])^2$,  
  Consequently, $\DS[i,j].U-\DSp[i,j].U \neq 0$.

  Therefore, for any cell $(i,j)$ with $1 < i < j$
  that lies on the top row of any character run in $A$,
  $\DS[i,j] \neq \DSp[i,j]$.
  Since each top row $\ITOP=Ik+1$~$(1 \leq I \leq M-1)$
  contains $n-(Ik+1)$ such cells,
  and since $n = kM$,
  there are $\sum_{I=1}^{M-1}{(n-(Ik+1))} =n(M-1)-k((M-1)M)/2-(M-1)  = (n-2)(M-1)/2$ such cells for top rows $I = 1, \ldots, M-1$.
  Symmetric arguments show that
  there are  $\sum_{J=1}^{N-1}{(m-(Jl+1))}=(m-2)(N-1)/2$ cells
  with $\DS[i,j] \neq \DSp[i,j]$ for left rows $J = 1, \ldots, N-1$.
  Thus, $\changed \geq ((n-2)(M-1)+(m-2)(N-1))/2 \in\Omega(mN+nM)$.
  \qed
\end{proof}

\subsection{Proof for Lemma~\ref{lem:updating_DP-table_worst-case}}

\begin{proof}
    The lemma can be shown in a similar manner to the proof for Theorem~\ref{theo:worst_case_changed} above.
  In so doing, we set $M = n$ and $N = n$ in the strings $A$ and $B$
  of Section~\ref{sec:evaluation},
  and consider strings $A$ and $B$ such that
  $A[i-1] > A[i]$ for $1 < i \leq m$
  and $B[j-1] < B[j]$ for $1 < j \leq n$, and $A[m] > B[n]$.
  Since deleting the leftmost character $B[1]$ of $B$ is symmetric to
  appending a new character $b_1$ to $B$ such that $b_1 < B[1]$,
  we get $\Omega(mn)$ lower bound for the number of cells
  where $D[i,j] \neq \Dp[i,j]$ per appended character.
  If we repeat this by recursively appending $k$ new characters $b_i$
  such that $b_{i} < b_{i-1} < \cdots < B[1]$ for $i = 2, \ldots, k$,
  we get $\Omega(mn)$ lower bound for the number of cells
  where $D[i,j] \neq \Dp[i,j]$ for each $b_i$.
  Hence there are a total of $\Theta(kmn)$ cells in $\Dp$ that differ from the
  corresponding cells in $\D$, for $k$ edit operations on $B$.
  \qed
\end{proof}

